\documentclass[letterpaper,12pt]{article}

\usepackage{epsfig,psfrag,amsmath,amssymb,latexsym}
\usepackage[usenames,dvipsnames]{color}
\usepackage{amscd, amsthm, amsfonts}
\usepackage{amsfonts}
\usepackage[top=2cm, bottom=2cm, left=2cm, right=2cm]{geometry}
\usepackage[usenames,dvipsnames]{xcolor}

\newtheorem{thm}{Theorem}[section]
\newtheorem{lemma}[thm]{Lemma}
\newtheorem{prop}[thm]{Proposition}
\newtheorem{cor}[thm]{Corollary}

\theoremstyle{remark}
\newtheorem{rmk}[thm]{Remark}

\def\cA{\mathcal{A}}
\def\cB{\mathcal{B}}
\def\cC{\mathcal{C}}

\def\eps{{\varepsilon}}

\def\tf{{\tilde{f}}}
\def\tF{{\widetilde{F}}}

\def\bW{{\mathbb{W}}}
\def\tW{{\widetilde{W}}}
\def\tbW{{\widetilde{\mathbb{W}}}}

\def\tE{{\widetilde{E}}}
\def\tSigma{{\widetilde{\Sigma}}}

\def\bE{{\mathbb{E}}}
\def\bP{{\mathbb{P}}}
\def\bR{{\mathbb{R}}}
\def\bT{{\mathbb{T}}}
\def\bX{{\mathbb{X}}}
\def\bY{{\mathbb{Y}}}
\def\bW{{\mathbb{W}}}

\def\hphi{{\hat{\phi}}}
\def\hn{{\hat{n}}}

\def\hW{{\hat{W}}}
\def\hbW{{\hat{\mathbb{W}}}}

\begin{document}

\title{Autonomous evolution of electron speeds in a thermostatted system:
exact results}
\date{31 August 2018 \quad (upd.\ 26 December 2018)}
\author{F. Bonetto,\\
\emph{\small School of Mathematics, Georgia Institute of
Technology, Atlanta GA 30332, USA} \\
N. Chernov\footnote{Deceased},\\
\emph{\small Department of Mathematics, University of Alabama at
Birmingham,
Birmingham AL 35294, USA} \\
A. Korepanov,\\
\emph{\small College of Engineering, Mathematics and Physical Sciences, 
University of Exeter, Exeter,
UK}\\
J.L. Lebowitz.\\
\emph{\small Departments of Mathematics and Physics, Rutgers
University, Piscataway NJ 08854, USA} }

\maketitle

\begin{abstract}
We investigate a dynamical system consisting of $N$ particles 
moving on a $d$-dimensional torus under the action of an electric 
field $E$ with a Gaussian thermostat to keep the total energy constant.
The particles are also subject to stochastic collisions which randomize
direction but do not change the speed. We prove that in the van Hove scaling limit,
$E\to 0$ and $t\to t / |E|^2$, the trajectory of the speeds $v_i$ is described by
a stochastic differential equation corresponding to diffusion on a 
constant energy sphere. This verifies previously conjectured behavior.

Our results are based on splitting the system's evolution into a ``slow'' process
and an independent ``noise''. We show that the noise, suitably rescaled,
converges to a Brownian motion, enhanced in the sense of rough paths.
Then we employ the It\^o-Lyons continuity theorem to identify the limit of the
slow process.
\end{abstract}

\section{Introduction}

Many-particle dynamical systems in which different quantities evolve on 
different time scales are common in nature.
This situation arises when the 
microscopic degrees of freedom come to some local equilibrium (stationary) 
state characterized by parameters which vary slowly in time. This leads to 
autonomous macroscopic equations, such as the Navier-Stokes equation,
with or without stochastic terms. Such equations are generally very hard
to derive rigorously, especially in situations which involve deterministic
dynamics.
In this note we continue our investigation of such a system of $N$
interacting particles with both deterministic and stochastic 
dynamics which conserves the total kinetic energy.

The system we consider is a variation of the 
Drude-Lorentz model of electrical 
conduction in a metal~\cite{AM}. It consists of $N$ particles
(electrons) in a $d$-dimensional torus under the action
of a constant external field $E$. The particles undergo 
elastic collisions with fixed or random scatterers,
which change directions of the velocities $\{p_k\}$
but not the speeds $\{|p_k|\}$, and are subject
to a Gaussian thermostat which keeps the total kinetic energy
of the system constant. The thermostat introduces dynamical interactions
between the particles.

We have studied this system extensively, in $d=2$,
via numerical simulations and approximate analytical methods,
using various models for the elastic 
collisions~\cite{BCKLs,BCKL,BDL,BDLR}. We have argued there 
that in all collision models, for a weak field $E$ the evolution
splits into a fast and slow parts which evolve essentially
independently, with the slow
part satisfying an autonomous diffusive equation.
We were however unable to prove this in a rigorous way: see 
Section~\ref{sec:RPV}.

Here we show, for the first time, in a rigorous mathematical way,
using the simplified collision model, that
the long time weak field evolution of the properly scaled system
(van Hove scaling) is described by an autonomous SDE, driven by an
$N$-dimensional Brownian motion.
To do this we apply the theory of ``rough paths'', pioneered by
Lyons (see e.g. \cite{Lyons98}).
We follow and adapt the approach of Kelly and Melbourne~\cite{KM16} for
nonuniformly hyperbolic dynamical systems. 
We find a way to decompose the evolution of the velocities $\{p_k\}$ into 
``fast'' and ``slow'' components, with the fast component uncoupled
from the slow.

\begin{rmk}
  A straightforward decomposition would be $p_k \mapsto (v_k, \omega_k)$,
  where $v_k = |p_k|$ and $\omega_k = p_k / |p_k|$. This
  corresponds to an intuition of quickly changing directions
  and slowly changing speeds. However, in such a decomposition,
  the evolution of $\{\omega_k\}$ depends on $\{v_k\}$: between collisions,
  the particles change directions, and the lower a particle's speed,
  the faster its direction changes.
  The influence of $\{v_k\}$ on $\{\omega_k\}$ causes substantial problems:
  the topic of rough differential equations with noise
  coupled to the solution is an unexplored area.
  (Though some related progress has been made in dynamical systems
  \cite{DeSimoiLiverani16,DeSimoiLiverani18,Dolgopyat05}.) We
  use a different decomposition, which may look technical and artificial
  but gives independent noise. 
\end{rmk}

To make our approach work, we make the following simplifying assumptions:
\begin{itemize}
  \item we assume that the collision rate of particle $k$
    is independent from its speed $v_k$. In the kinetic theory of gases,
    these are referred to as ``Maxwellian'' collisions;
  \item at each collision, the outgoing direction $\omega$
    of the colliding particle is selected uniformly on the unit
    sphere $S^{d-1}$ in $\bR^d$.
\end{itemize}

The paper is organized as follows. In Section~\ref{sec:MR} we
give a precise description of the model and state our results.
Section~\ref{sec:pf-main} contains the main idea and strategy 
of the proof of our main result while Section~\ref{sec:iwip}
contains a technical adaptation of \cite{KM16} to our situation,
used in Section~\ref{sec:pf-main}.
Finally, in Section~\ref{sec:RPV}
we discuss the connection of this work with 
previous work in this ongoing research.
In Appendix~\ref{sec:CDESM} we give a standard example where the solution
of a differential equation is not a continuous function of a
driving signal, and in Appendix~\ref{sec:PSD} we show that
an $n$-particle diffusion on the sphere of radius $\sqrt{n}$
corresponds, as $n \to \infty$, to an Ornstein-Uhlenbeck process
for the motion of a single particle.

\section{Model and Results}
\label{sec:MR}

We consider a system formed by \(N\) particles moving on a torus 
\(\bT^d\), \(d \geq 2\),\footnote{See Remark~\ref{rmk:1d} for $d=1$.}
with positions \(q_k\) and velocities \(p_k = \dot{q}_k\). The particles move under 
an electric field \(E\) and a Gaussian thermostat, which preserves the total energy,
so \(U = \sum_k |p_k|^2\) is constant~\cite{BCKLs,BCKL,BDL,BDLR}:
\begin{equation} \begin{aligned}
  \label{eq:gt0}
  \begin{cases}
    d q_k &= p_k \, dt, \\
    dp_k &= \Bigl( E - \frac{\sum_j E \cdot p_j}{U} p_k \Bigr) \, dt,
  \end{cases}
\end{aligned} \end{equation}
In addition each particle experiences random collisions, independent of the other particles.
Collisions are driven by Poisson processes with rate \(\lambda > 0\), the same
for all particles. At a collision, a particle's direction $p_k/|p_k|$,
which is a point on the unit sphere \(S^{d-1}\) in \(\bR^d\), changes
randomly and uniformly on \(S^{d-1}\), while the speed \(|p_k|\)
is preserved.

Suppose that \(E = \eps \hn\), where \(\hn \in \bR^d\) is a fixed unit 
vector, and \(\eps > 0\). Let \(v_k = |p_k|\).
Let the initial conditions \(p_k(0)\) and \(q_k(0)\) be fixed.
Then \(v = (v_1, \ldots, v_N)\) is random process with continuous sample paths.

\begin{rmk}
  In our model, positions of particles affect neither velocities between
  collisions nor collisions themselves. Working with the velocities,
  the positions can be safely ignored.
\end{rmk}

Heuristic arguments~\cite{BCKLs} show that for small \(\eps\), the time changed
process \(v^\eps\), \(v^\eps(t) = v(\eps^{-2} t)\), behaves like a nontrivial stochastic process.
In this paper, we describe this behavior rigorously.

Our main result is:

\begin{thm}
  \label{thm:main}
  As \(\eps \to 0\), the time changed processes \(v^\eps\),
  \(v^\eps(t) = v(\eps^{-2} t)\), converge weakly to the solution \(v^0\)
  of the It\^o stochastic differential equation
  \begin{equation}
    \label{SDE}
    \begin{aligned}
      d v^0_k
      & = \delta^{1/2}\Bigl[ d W_k - v^0_k \frac{\sum_j v^0_j \, d W_j}{U} \Bigr]
      + \delta \Bigl[\frac{d-1}{2v^0_k}
      - \frac{(Nd-1) v ^0_k}{2 U}
      \Bigr] \, dt
      , \\
      v^0_k(0) &= |p_k(0)|
      .
    \end{aligned}
  \end{equation}
  Here \(\delta = 2 \lambda^{-1} d^{-1}\) and
  \(W_1, \ldots, W_N\) are standard independent Brownian motions.
  The weak convergence is in the \(C^0([0,\infty), \bR^N)\) topology.
\end{thm}

\begin{rmk}
  Since the processes \(v^\eps\) have continuous sample paths in \(\bR^N\),
  it is natural to think of them as random elements of the space
  \(C^0([0,\infty), \bR^{N})\) of continuous functions
  with the usual topology of uniform convergence on compact sets.
  
  The weak convergence of \(v^\eps\) to \(v^0\) in \(C^0([0,\infty), \bR^N)\) topology
  means that \(\bE f(v^\eps) \to \bE f(v^0)\) for every continuous
  test function \(f \colon C^0([0,\infty), \bR^N) \to \bR\).
  Weak convergence is also known as convergence in distribution
  and convergence in law.
\end{rmk}

\begin{rmk}
  \label{rmk:sSDE}
  The evolution of speeds $\{v^0_k\}$ can be described by a much simpler SDE.
  There exists a $dN$-dimensional process $u^0 = (u^0_1, \ldots, u^0_N)$, with
  each $u^0_k$ a $d$-dimensional process and $v^0_k = |u^0_k|$ at all times.
  The process $u^0$ is a diffusion on the sphere \(|u^0|^2 = U\) which can be
  written as a solution of a Stratonovich SDE
  \begin{equation}
    \label{eq:dSDE}
    d u^0
    = \Bigl( I - \frac{u^0 (u^0)^*}{U} \Bigr) \, \circ d W
    .
  \end{equation}
  Here $W$ is a $dN$-dimensional
  Brownian motion with covariance matrix $2 \lambda^{-1} d^{-1} I$.
  See Section~\ref{sec:pf-main}, specifically Theorem~\ref{thm:u},
  for details.
\end{rmk}

\begin{rmk}
  \label{rmk:1d}
  Theorem~\ref{thm:main} is restricted to $d \geq 2$: with $d=1$,
  SDE~\eqref{SDE} gives a wrong process which allows $v_k$ to become negative.
  Remark~\ref{rmk:sSDE}, however, fully applies to all $d \geq 1$.
  In our proof, Theorem~\ref{thm:main} is derived from Remark~\ref{rmk:sSDE}
  and the restriction $d \geq 2$ comes up in the transition 
  from~\eqref{eq:dSDE} to~\eqref{SDE}.
  With $d=1$, the correct SDE for $v_k$ is expected to be more complicated
  than~\eqref{SDE}, much like the SDE for the absolute value
  of a one-dimensional Brownian motion.
\end{rmk}

\section{Proof of Theorem~\ref{thm:main}}
\label{sec:pf-main}

\subsection{Strategy}

We split the velocities \(p_k\) into ``slow'' and ``fast''
components, such that the fast component is independent from the slow one.
Then we write the evolution of the slow component as a \emph{rough differential equation,}
where the noise is generated by the fast component alone.
We show that the noise, suitably rescaled, converges  as \(\eps \to 0\) to an
\emph{enhanced Brownian motion} in a suitable \emph{rough path topology.}
Then we use the It\^o-Lyons continuity theorem to describe the limiting slow process.
Finally, we recover the speeds \(v_k\) from the slow process.

\subsection{Decomposition into fast and slow components}

We start by (re)defining the random collisions in a convenient way.
For \(1 \leq k \leq N\), let 
\(\tau_k = (\tau_k^1, \tau_k^2, \ldots)\)
be Poisson processes with rate \(\lambda\), and let
\(g_k^1, g_k^2, \ldots\) be random matrices
in the orthogonal group \(O(d)\), chosen uniformly
(i.e.\ with respect to the Haar measure).
We assume that all these are mutually independent.

We make the \(k\)-th particle collide at times \(\tau_k^n\),
with the instantaneous change of velocity at time
\(\tau_k^n\) given by \(p_k (\tau_k^n) = g_k^n \, p_k (\tau_k^n - 0)\),
where \(p_k (t - 0)\) stands for the left limit of \(p_k\) at
time \(t\).

Define \(\phi_k \colon [0,\infty) \to O(d)\), \(k=1\ldots N\) by 
\(\phi_k(t) = I\) for \(t < \tau_k^1\) and
\[
  \phi_k (t) = g_k^{n} g_k^{n-1} \cdots g_k^1
  \qquad \text{for} \qquad t \in [\tau_k^{n}, \tau_k^{n+1})
  , \quad n \geq 1
  .
\]

Let \(u_k = \phi_k^* p_k\), where the star means the transpose
(also inverse in \(O(d)\)). Then \(u_k(0) = p_k(0)\)
and \(|u_k| = |p_k| = v_k\) at all times. Observe that
\(u_k\) does not jump at collisions, in contrast with \(p_k\).
From the equations of motion~\eqref{eq:gt0} we obtain
\begin{equation}
  \label{eq:agq}
  du_k = \Bigl( \phi_k^* E - u_k \frac{\sum_j u_j^* \phi_j^* E}{U} \Bigr) \, dt.
\end{equation}
Note that \(u_k\) are continuous and piecewise smooth processes.

Let \(\Phi_k (t) = \int_{0}^{t} \phi_k^*(s) \hn \, ds\).
Writing all \(u_k\), \(1 \leq k \leq N\) as one \(N d\)-dimensional vector \(u\), 
and similarly all \(\Phi_k\) as one vector \(\Phi\), we 
rewrite~\eqref{eq:agq} as
\begin{equation}
  \label{eq:azmny}
  du = \eps A(u) \, d \Phi,
\end{equation}
where \(A(u) = I - \frac{u u^*}{U}\) is a smooth matrix
valued function.

Define \(u^\eps (t) = u(\eps^{-2} t)\) and \(W^\eps (t) = \eps \Phi (\eps^{-2} t)\).
For every \(\eps > 0\), the process \(W^\eps\) is piecewise smooth and
\(u^\eps\) is the solution of an ordinary differential equation
\begin{equation}
  \label{eq:jahsd}
  d u^\eps = A(u^\eps) \, dW^\eps
  , \qquad
  u^\eps(0) = \xi
  ,
\end{equation}
where \(\xi \in \bR^{Nd}\) corresponds to the initial condition \(u^\eps_k(0) = p_k(0)\).

The above is the decomposition into a ``slow'' component
\(u\) and an uncoupled ``fast'' component \(\Phi\). While \(u\) has
no obvious physical meaning, we can still use it to recover the speeds,
because \(v_k = |u_k|\).

\subsection{Limiting behavior}
\label{sec:LB}

Our next goal is to identify the limit of \(u^\eps\) as \(\eps \to 0\).
We can treat \(u^\eps\) as a function of \(W^\eps\). Indeed, there is a 
\emph{solution map} \(\Gamma\), defined on all piecewise smooth paths,
such that \(\Gamma(W^\eps) = u^\eps\) for each \(\eps > 0\).

We will show that \(W^\eps\) converges to a 
\(dN\)-dimensional Brownian motion \(W\) with 
covariance matrix \(2 \lambda^{-1} d^{-1} I\).
(Convergence is weak in the \(C^0([0,\infty), \bR^d)\) topology.)

Then, heuristically, one would expect that \(u^0 = \lim_{\eps \to 0} \Gamma(W^\eps) = \Gamma(W)\).
Such a statement requires continuity of \(\Gamma\) on a suitable space of paths, including
all \(W^\eps\) and \(W\).
There are two immediate problems:
\begin{itemize}
  \item As a Brownian motion, \(W\) is rather irregular,
    so \(\Gamma(W)\) cannot be understood as a solution of an ordinary differential equation.
    It needs an interpretation, possibly as a solution of the stochastic differential equation
    \(d u^0 = A(u^0) \, \diamond dW\), where \(\diamond\) means integration in the sense of
    e.g.\ It\^o, Stratonovich or backward It\^o.
  \item There is no reason for \(\Gamma\) to be continuous.
    In fact, no matter what interpretation of \(\Gamma(W)\) we choose,
    \(\Gamma\) will fail to be continuous in any usable way.
    (We provide a standard example in Appendix~\ref{sec:CDESM}.)
\end{itemize}

Identifying \(u^0\) is a standard problem in the theory of rough paths.
We fix \(T > 0\) and consider all processes in the time interval \([0,T]\).
For \(s, t \in [0,T]\), let
\begin{align*}
  W^\eps(s,t) & = W^\eps(t) - W^\eps(s),
  &
  \bW^\eps(s,t) &= \int_{s}^{t} W^\eps(r) \otimes d W^\eps(r)
  .
\end{align*}
Here \(\otimes\) denotes the tensor product, i.e.\ if \(A, B \in \bR^n\),
then \(A \otimes B \in \bR^{n \times n}\) is given by \((A \otimes B)_{j,k}
= A_j B_k\).

The pairs \((W^\eps, \bW^\eps)\) are the \emph{canonical lifts} of
the original piecewise smooth paths \(W^\eps\), see
\cite[Section~2]{FrizHairer14}. These pairs belong to the space \(\cC^\alpha\) of
\(\alpha\)-H\"older rough paths with \(\alpha \in (1/3, 1/2)\).

Fix \(\alpha \in (1/3, 1/2)\).
The space \(\cC^\alpha\) consists of pairs 
\((X, \bX) \in C^0([0,T] \times [0,T], \bR^{Nd} \times \bR^{Nd \times Nd})\)
with
\[
  \|(X, \bX)\|_\alpha
  = |X|_\alpha + \sqrt{|\bX|_{2\alpha}}
  < \infty,
\]
where
\[
  |X|_\alpha
  = \sup_{s \neq t \in [0,T]} \frac{|X(s,t)|}{|t-s|^\alpha}
  \qquad \text{and} \qquad
  |\bX|_{2\alpha}
  = \sup_{s \neq t \in [0,T]} \frac{|\bX(s,t)|}{|t-s|^{2\alpha}}
  .
\]
Further we require that $(X, \bX)$ satisfies \(X(t,t) = 0\) and \(\bX(t,t) = 0\) for all
\(t\), and
\[
  \bX_{s,t} - \bX_{s,u} - \bX_{u,t}
  = X_{s,u} \otimes X_{u,t}
  .
\]
The above is known as the ``Chen's relation''. It holds for and is inspired by
the canonical lifts of smooth paths such as $(W^\eps, \bW^\eps)$.
Because of the Chen's relation, \(\cC^\alpha\) is not a linear subspace
of \(C^0\).

The topology on \(\cC^\alpha\) is inherited from the seminorm \(\| \cdot \|_\alpha\) on
\(C^0\): it is given by the distance
\[
  d_{\cC^\alpha} ((X, \bX), (Y, \bY))
  = \sup_{s \neq t \in [0,T]} \frac{|X(s,t) - Y(s,t)|}{|t-s|^\alpha}
  + \sup_{s \neq t \in [0,T]} \frac{|\bX(s,t) - \bY(s,t)|}{|t-s|^{2\alpha}}
  .
\]

A key result in rough paths is 
the It\^o-Lyons continuity theorem~\cite[Section~8]{FrizHairer14}.
Applied to the differential equation~\eqref{eq:jahsd},
it gives a \emph{continuous} map
\(\Gamma \colon \cC^\alpha \to C^0([0,T], \bR^{Nd})\), 
such that \(\Gamma(W^\eps, \bW^\eps) = u^\eps\) for all \(\eps\).

\begin{rmk}
  Our application of the It\^o-Lyons continuity theorem requires that the function
  $A(u)$ in~\eqref{eq:jahsd} has bounded continuous derivatives up to the second order.
  Our function \(A(u) = I - \frac{u u^*}{U}\) is not bounded on \(\bR^{Nd}\).
  However, all the solutions $u^\eps$ are uniformly bounded, namely restricted to
  the sphere $|u^\eps|^2 = U$, where the global unboundedness of $A(u)$ or
  its derivatives does not cause problems.
\end{rmk}

Suppose now that the random elements \((W^\eps, \bW^\eps)\) converge weakly
in \(\cC^\alpha\) to some \((W, \bW)\). Then, using the continuous mapping
theorem, we find \(u^0 = \Gamma(W, \bW)\). 

Further, in Section~\ref{sec:iwip}, we show that \((W^\eps, \bW^\eps)\) do indeed converge to a
\(\cC^\alpha\)-valued random process \((W, \bW)\), where \(W\) is the Brownian motion
with covariance matrix \(2 \lambda^{-1} d^{-1} I\) and,
using \(\circ dW\) to denote Stratonovich integration,
\begin{align*}
  W(s,t) &= W(t) - W(s),
  &
  \bW(s,t) &= \int_s^t W(r) \otimes \circ dW(r)
  .
\end{align*}

\begin{rmk}
  A rough path $(W, \bW)$, where $W$ is a Brownian motion,
  is often referred to as \emph{enhanced Brownian motion.}
  In our case, the enhancement is Stratonovich. Often it is
  natural to consider the It\^o enhancement; in general,
  the options for enhancement are plentiful.
\end{rmk}

Rough integration against \((W,\bW)\) coincides with Stratonovich integration
against \(W\) (see~\cite{FrizHairer14}).
This means that \(u^0 = \Gamma(W, \bW)\)
is the solution of the Stratonovich stochastic differential equation
\(du^0 = A(u^0) \, \circ dW\).
The convergence is on the time interval \([0,T]\), but \(T\) is arbitrary, so convergence on
\([0, \infty)\) follows.

In other words, we obtain the following result:

\begin{thm}
  \label{thm:u}
  As \(\eps \to 0\), the processes \(u^\eps\) converge weakly
  in the \(C^0([0,\infty), \bR^{Nd})\) topology to \(u^0\), the
  solution of Stratonovich differential equation
  \begin{equation}
    \label{eq:jaao}
    d u^0
    = \Bigl( I - \frac{u^0 (u^0)^*}{U} \Bigr) \, \circ d W
    , \qquad u^0(0) = \xi
    .
  \end{equation}
  Here \(W\) is a \(dN\)-dimensional Brownian motion with covariance matrix
  \(2 \lambda^{-1} d^{-1} I\).
\end{thm}

\begin{rmk}
  The process \eqref{eq:jaao} is a diffusion on the sphere
  \(|u^0|^2 = U\). When $Nd$ is large and $U = Nd$, coordinate projections of
  $u^0$ are close to an Ornstein-Uhnelbeck process. We provide the details in
  Appendix~\ref{sec:PSD}.
\end{rmk}

\begin{rmk}
  Theorem~\ref{thm:main} is a corollary of Theorem~\ref{thm:u}. 
\end{rmk}


\section{Convergence of rough paths}
\label{sec:iwip}

In Section~\ref{sec:pf-main} we introduced the random rough paths 
\((W^\eps, \bW^\eps)\) and \((W, \bW)\). 
To prove Theorem~\ref{thm:u},
it remains to show that \((W^\eps, \bW^\eps)\) converge weakly to
\((W, \bW)\) in the space \(\cC^\alpha\) of \(\alpha\)-H\"older rough paths,
with $\alpha \in (1/3, 1/2)$. This is the goal of this section.

We follow the arguments for deterministic dynamical systems by
Kelly and Melbourne~\cite{KM16}. Our situation is simpler
because of randomness, but \cite{KM16} does not cover it, so we
provide an adaption of their proof.

\begin{rmk}
  Writing $W^\eps(t) = W^\eps(0,t)$ and $\bW^\eps(t) = \bW^\eps(0,t)$,
  we can work with the rough paths \((W^\eps, \bW^\eps)\) as a random processes,
  defined on $t \in [0,\infty)$. The ``increments'' $W^\eps(s,t)$ and $\bW^\eps(s,t)$
  can be recovered from the Chen's relation:
  \begin{align*}
    W^\eps(s,t) & = W^\eps(t) - W^\eps(s)
    , \\
    \bW^\eps(s,t) &= \bW^\eps(t) - \bW^\eps(s) 
      - W^\eps(s) \otimes W^\eps(t)
    .
  \end{align*}
\end{rmk}

To prove convergence of \((W^\eps, \bW^\eps)\) to \((W, \bW)\)
in the $\cC^\alpha$ topology, it is sufficient to show
convergence in the weaker uniform topology together with
suitable moment bounds (c.f.~\cite[Theorem~9.1]{KM16}):

\begin{lemma}
  \label{lm:prer}
  As $\eps$ goes to zero, the process
  $(W^\eps, \bW^\eps)$ converges weakly to
  $(W, \bW)$ in $C^0([0, \infty), \bR^{Nd} \times \bR^{Nd \times Nd})$.
  
  Moreover, there exist $q>3$ and $C>0$ such that
  for all $\eps > 0$ and $0 \leq s \leq t$,
  \begin{equation} \label{eq:anya}
    \|W^\eps (s,t)\|_{2q} \leq C |t-s|^{1/2}
    \qquad \text{and} \qquad
    \|\bW^\eps (s,t)\|_q \leq C |t-s|
    .
  \end{equation}
  Here \(\| X \|_q = (\bE |X|^q )^{1/q}\) is the \(L^q\) norm.
\end{lemma}

Here is an outline of the proof of Lemma~\ref{lm:prer}:
\begin{itemize}
  \item Observe that the processes $W^\eps(t)$ do not have stationary increments.
    This comes from the fact that $\phi_k(0) = I$ instead of being distributed
    uniformly in $O(d)$. This is a minor inconvenience which we solve by introducing
    a random time shift $\tau$ such that the random processes
    $\hW^\eps(t) = W^\eps (\eps^2 \tau + t) - W^\eps(\eps^2 \tau)$
    do have stationary increments.
    We show that the rough paths $(\hW^\eps, \hbW^\eps)$, where $\hbW^\eps$
    are the respective iterated integrals, well approximate
    $(W^\eps, \bW^\eps)$.
    
  \item We represent $\hW^\eps(t) = \eps \int_0^{\eps^{-2} t} h \circ F^s \, ds$,
    where $F^t\colon \Omega \to \Omega$ is a measure
    preserving semiflow on a probability space
    $(\Omega, \bP)$ and $h \colon \Omega \to \bR^{Nd}$ is an observable.
    
  \item We consider a discrete time dynamical system $F \colon \Omega \to \Omega$,
    where $F = F^1$. It preserves the measure $\bP$, and our construction ensures 
    that $F$ is mixing.
    The semiflow $F^t$ is a suspension over $F$ with the roof function equal to $1$.
    We consider the induced observable \(V \colon \Omega \to \bR^{Nd}\),
    \(V = \int_{0}^{1} h \circ F^t \, dt\), and use a martingale-coboundary
    decomposition $V = m + \chi \circ F - \chi$, where both $\chi$ and $m$ are bounded
    and $m$ is a ``martingale part''. 
    This means that for every $n$, the ``backward'' sum $\sum_{j=n-k}^n m \circ F^j$ is a martingale
    on $k=0,\ldots, n$.
    
  \item Let 
    \(
      \tW^\eps (t) 
      = \eps \sum_{j=0}^{\lfloor \eps^{-2} t \rfloor} V \circ F^j
    \) be a discrete time version of $\hW$, and 
    let $\tbW^\eps$ be the corresponding
    iterated integral.
    Then $(\tW, \tbW)$ is a random c\`adl\`ag process.
    By \cite[Theorem~4.3]{KM16}, if $F$ is mixing and $V$ allows
    a martingale-coboundary decomposition as above, then 
    the weak limit of $(\tW^\eps, \tbW^\eps)$ in the $C^0$ topology
    is described by Green-Kubo-like formulas~\eqref{eq:tWlim}.
    
  \item The processes $(\tW^\eps, \tbW^\eps)$, $(\hW^\eps, \hbW^\eps)$
    and $(W^\eps, \bW^\eps)$ are closely
    related, and knowing the weak limit of the first allows us to compute
    the weak limit of the others.
    
  \item The moment bounds for $(\tW^\eps, \tbW^\eps)$ and
    $(\hW^\eps, \hbW^\eps)$, and hence for $(W^\eps, \bW^\eps)$,
    are implications of the martingale-coboundary decomposition~\cite[Section~7]{KM16}.
\end{itemize}

In the rest of this section, we implement the above.

\subsection{Probability measure preserving flow}

Note that the processes \(\phi_k\) are not stationary. For instance,
\(\phi_k(0) = I\) for all \(k\). Let \(\tau = \max \{\tau_k^1 \colon 1 \leq k \leq N \}\)
and \(\hphi_k(t) = \phi_k(t+\tau)\), \(t \geq 0\).
Now, \(\hphi_k\) are stationary processes, and so are
\(\hphi_k^* \hn\).

Let \(\psi\) be a random process with values in \(\bR^{Nd}\), obtained by stacking 
together all the coordinates of \(\hphi_k^* \hn\), \(k=1\ldots N\).

Let \(\Omega = D([0,\infty), \bR^{Nd})\) be the space of c\`adl\`ag functions.

Let \(\bP\) be the probability measure on
\(\Omega\), corresponding to the distribution of 
\(\psi\), and let \(\bE\) denote the corresponding expectation.

Define the flow \(F^t \colon \Omega \to \Omega\) by \((F^t x)(s) = x(t+s)\)
for \(s,t \geq 0\),
and let $h \colon \Omega \to \bR^{Nd}$,
$h(x) = x(0)$ be an observable. 
Since \(\psi\) is a stationary process, the measure \(\bP\) is \(F^t\)-invariant.

Define \(\hW^\eps\) and \(\hbW^\eps\) by
\[
  \hW^\eps(t) = \eps \int_0^{\eps^{-2} t} h \circ F^s \, ds
  \qquad \text{and} \qquad
  \hbW^\eps(t) = \int_0^t \hW^\eps (s) \otimes d\hW^\eps (s)
  .
\]

\begin{rmk}
  \label{rmk:wic}
  Where it is convenient, we assume that \(W^\eps\), \(\hW^\eps\) and \(\tau\)
  are defined on the same probability space such that
  \(W^\eps (\eps^2 \tau + t) - W^\eps(\eps^2 \tau) = \hW^\eps(t)\) for all \(t \geq 0\).
  The iterated integrals \(\bW^\eps\) and \(\hbW^\eps\) are fully determined by
  \(W^\eps\) and \(\hW^\eps\), so they also belong to this probability space.
\end{rmk}

\subsection{Discrete time system}
\label{sec:dts}

Let \(F = F^1\), and \(P \colon L^1(\Omega) \to L^1(\Omega)\) be the
(Ruelle-Perron-Frobenius) \emph{transfer operator}, corresponding to \(F\) and \(\bP\).
Formally, \(P\) is defined by 
\[
  \bE (Pv \, w) = \bE(v \, w \circ F)
  \quad \text{ for all }
  v \in L^1(\Omega)
  \text{ and } 
  w \in L^\infty(\Omega)
  .
\]

\begin{rmk}
  \label{rmk:nnao}
  \(P v\) can be computed explicitly.
  For \(x \in \Omega\),
  \[
    (P v)(x)
    = \bE( v \mid F = x)
    = \int_\Omega v (y) \, d\bP( y \mid Fy = x)
    ,
  \]
  where \(\bP(\cdot \mid F = x)\) is the regular 
  conditional probability
  corresponding to the \emph{observable} \(F \colon \Omega \to \Omega\),
  and \(\bE( \cdot \mid F=x)\) 
  is the corresponding expectation.
  (See \cite{CP97} for a guide on guilt-free manipulation
  with regular conditional probabilities.)
  
  Similarly,
  \(
    (P^k v)(x)
    = \bE( v \mid F^k = x)
  \)
  for \(k \geq 1\).
\end{rmk}

Define \(V \colon \Omega \to \bR^{Nd}\) by 
\(V = \int_{0}^{1} h \circ F^t \, dt\).
Then \(V(x) = \int_{0}^{1} x(t) \, dt\)
for \(x \in \Omega\).
Note that \(V\) is a bounded random variable, and due to the
natural symmetries of our model, \(\bE V = 0\).

\begin{prop} For \(a \in \bR^{Nd}\), 
  \label{prop:mlg}
  \[
    \int_\Omega x(t) \, d\bP(x \mid x(0) = a)
    = e^{- \lambda t} a.
  \]
\end{prop}

\begin{proof}
  Recall the definition of \(\phi_k\). Note that for \(a' \in O(d)\),
  \begin{align*}
    \bE(\phi_k^*(t) \hn \mid \phi_k(0) = a', \ t < \tau_k^1) & = a'^* \hn,
    \quad \text{and}
    \\
    \bE(\phi_k^*(t) \hn \mid  \phi_k(0) = a', \ t \geq \tau_k^1) & = 0.
  \end{align*}
  Hence
  \[
    \bE( \phi_k^* \hn \mid \phi_k(0) = a')
     = \bP(t < \tau_k^1) \, a'^* \hn
     = e^{-\lambda t} a'^* \hn
     .
  \]
  The result follows.
\end{proof}

\begin{prop} \label{prop:hqst}
  For \(x \in \Omega\) and \(k \geq 0\),
  \[
  (P^k V)(x) 
  = \frac{e^{-\lambda k}(e^{\lambda} - 1)}{\lambda}
    x(0). 
  \]
\end{prop}

\begin{proof}
  We observe that \((\Omega, \bP)\) is probability space of Markov, stationary
  and time-reversible processes. In particular, for \(x \in \Omega\),
  \(t \in [0,1]\) and \(k \geq 1\),
  \begin{equation}
    \label{eq:an4r}
    \begin{aligned}
      & \int_\Omega y(t) \, d\bP(y \mid F^k y = x)
      = \int_\Omega y(t) \, d\bP(y \mid y(k) = x(0))
      \\ & = \int_\Omega y(k-t) \, d\bP(y \mid y(0) = x(0))
      = e^{-\lambda (k-t)} x(0)
      ,
    \end{aligned}
  \end{equation}
  where in the last step we used Proposition~\ref{prop:mlg}.
  
  Using \eqref{eq:an4r} and Fubini's theorem, write for \(k \geq 1\):
  \begin{align*}
    (P^k V) (x) 
    & = \bE( V \mid F^k = x )
    = \int_\Omega \int_{0}^{1} y(t) \, dt \, d\bP(y \mid F^k y = x)
    \\ & = \int_{0}^{1} \int_\Omega y(t) \, d\bP(y \mid F^k y = x) \, dt 
    = \int_{0}^{1} e^{-\lambda (k-t)} x(0) \, dt 
    = \frac{e^{- \lambda k}(e^{\lambda} - 1)}{\lambda} x(0)
    .
  \end{align*}
\end{proof}

Now we approximate \(V\) by a martingale, following
\cite[Section~4]{KM16}.
Define \(\chi, m \colon \Omega \to \bR^{Nd}\) by
\[
  \chi = \sum_{k=1}^{\infty} P^k V
  \qquad \text{and} \qquad
  V = m + \chi \circ F - \chi.
\]
Using Proposition~\ref{prop:hqst} and the definition of \(V\),
we compute \(\chi\) and \(m\) explicitly: 
\[
  \chi(x) = \frac{1}{\lambda} x(0)
  \qquad \text{and} \qquad
  m(x) = \int_{0}^{1} x(t) \, dt + \frac{x(0)- x(1)}{\lambda}
  .
\]
Clearly \(m,V \in L^\infty (\Omega)\). It is standard that
\(Pm = 0\) (see~\cite[Proposition~4.4]{KM16}).

Let \(V_i\) and \(x_i\) denote the \(i\)-th coordinates of \(V\) and \(x\)
respectively.

\begin{prop}
  \label{prop:najj}
  \[
    \int_\Omega x_i(t) x_j(s) \, d\bP(x)
    = 
    \begin{cases}
      0, & i \neq j \\
      d^{-1} e^{-\lambda |t-s|}, & i = j
    \end{cases}
    .
  \]
\end{prop}

\begin{proof}
  Let \(r \geq 0\), and let \(x\) be distributed in \(\Omega\) according to \(\bP\).
  Note that \(\sum_{j=1}^{Nd} x_j^2(r) = N\). Then due to the symmetry of the
  distribution \(\bP\),
  \[
    \int_\Omega x_j^2(r) \, d\bP(x) = d^{-1}
    \quad \text{for every } j
    . 
  \]
  Fix \(j\). By Proposition~\ref{prop:mlg},
  \[
    \int_\Omega x_j(0) x_j(r) \, d\bP(x \mid x(0))
    = x_j^2(0) e^{-\lambda r}
    .
  \]

  Without loss of generality suppose that \(s \leq t\).
  By the above, and using the fact that the measure \(\bP\) is stationary,
  write
  \begin{align*}
    \int_\Omega x_j(t) x_j(s) \, d\bP(x)
    & = \int_\Omega x_j(0) x_j(t-s) \, d\bP(x)
    \\ & = \int_\Omega \int_\Omega x_j(0) x_j(t-s) \, d\bP(x \mid x(0)) \, d\bP(x)
    \\ & = \int_\Omega x_j^2(0) e^{-\lambda (t-s)} \, d\bP(x)
    = d^{-1} e^{-\lambda (t-s)}
    .
  \end{align*}

  It remains to show that \(\int_\Omega x_i(t) x_j(s) \, d\bP(x) = 0\) when
  \(i \neq j\).
  For this we use again a symmetry of the distribution $\bP$: for each \(j\), it is invariant
  under transformation \(x_j \mapsto - x_j\).
  Thus with \(i \neq j\),
  \[
    \int_\Omega x_i(t) x_j(s) \, d\bP(x)
    = - \int_\Omega x_i(t) x_j(s) \, d\bP(x)
    = 0
    .
  \]
\end{proof}

\begin{prop} If \(i \neq j\), then \(\bE(V_i \, V_j \circ F^k) = 0\).
  Also,
  \[
    \bE(V_i \, V_i \circ F^k)
    = 
    \begin{cases}
    \frac{(e^{\lambda} - 1)^2 e^{-\lambda (k+1)}}{\lambda^{2} d} 
      , & k \geq 1 \\
    \frac{2(e^{-\lambda} + \lambda - 1)}{\lambda^2 d}
      , & k = 0
    \end{cases}
    .
  \]
\end{prop}

\begin{proof}
  Use Fubini's theorem and Proposition~\ref{prop:najj}:
  \begin{align*}
    \bE(V_i \, V_j \circ F^k)
    & = \int_\Omega 
      \int_{0}^{1} x_i(t) \, dt 
      \int_{k}^{k+1} x_j(s) \, ds
      \, d\bP(x)
    \\ & = \int_{0}^{1} \int_{k}^{k+1}
      \int_{\Omega} x_i(t) x_j(s) \, d\bP(x) \, ds \, dt
    = \int_{0}^{1} \int_{k}^{k+1}
      d^{-1} e^{-\lambda |s-t|} \, ds \, dt
    .
  \end{align*}
  The result follows.
\end{proof}

Define
\begin{align*}
  \tE & = \sum_{k=1}^{\infty} \bE ( V \otimes V \circ F^k)
  ,
  \\ 
  \tSigma & = \bE(V \otimes V) + \sum_{k=1}^{\infty}
    \bE ( V \otimes V \circ F^k + V \circ F^k \otimes V)
  .
\end{align*}
\begin{cor}
  \label{cor:nnja}
  \(\tE = \frac{1- e^{- \lambda}}{\lambda^2 d} I\)
  and \(\tSigma = \frac{2}{\lambda d} I\). 
\end{cor}

Define a c\`adl\`ag process
\(
  (\tW^\eps, \tbW^\eps) \colon \Omega \to D([0,\infty), \bR^{Nd} \times
  \bR^{Nd \times Nd} )
\)
by
\[
  \tW^\eps (t) 
  = \eps \sum_{j=0}^{\lfloor \eps^{-2} t \rfloor} V \circ F^j
  \qquad \text{and} \qquad
  \tbW^\eps (t) = \int_{0}^{t} \tW^\eps(s) \otimes d\tW^\eps(s)
  .
\]

By~\cite[Theorem~4.3]{KM16}, \((\tW^\eps, \tbW^\eps)\) converges weakly
to \((\tW, \tbW)\) in 
\(D([0,\infty),\bR^{Nd} \times \bR^{Nd \times Nd})\) in the uniform topology,
where \(\tW\) is the \(N\)-dimensional 
Brownian motion with covariance matrix \(\tSigma\),
and
\begin{equation}
  \label{eq:tWlim}
  \tbW (t) = \int_{0}^{t} \tW \otimes d\tW + \tE t
  .
\end{equation}

\subsection{Continuous time system}
\label{sec:cts}

In this part of the proof we are closely following~\cite[Section~6]{KM16}.

Define \(H \colon \Omega \times [0,1) \to \bR^{Nd}\), 
\(H(x,r) = \int_{0}^{r} x(t) \, dt\).
Let 
\[
  E = \tE + 
    \int_\Omega \int_0^1
      H(x,r) \otimes x(r)
    \, dr \, d \bP(x)
  .
\]

\begin{prop}
  \label{prop:aaji}
  \(
    E = \frac{1}{\lambda d} I
  \).
\end{prop}

\begin{proof}
  Indeed, using the definition of \(H\), Fubini's theorem
  and Proposition~\ref{prop:najj}, write 
  \begin{align*}
    &\int_\Omega \int_0^1
      H(x,r) \otimes x(r)
    \, dr \, d \bP(x)
    =
    \int_\Omega \int_0^1 \int_0^r 
       x(t) \, dt \otimes x(r) 
    \, dr \, d \bP(x)
    \\ & = 
    \int_0^1 \int_0^r \int_\Omega
       x(t) \otimes x(r) 
    \, d \bP(x) \, dt \, dr
    = 
    \int_0^1 \int_0^r
       d^{-1} e^{-\lambda |t-r|} I
    \, dt \, dr
    = \frac{e^{-\lambda} + \lambda - 1}{\lambda^2 d} I
    .
  \end{align*}
  The result follows from Corollary~\ref{cor:nnja}.
\end{proof}

Recall the definition of \(\hW^\eps\) and \(\hbW^\eps\).
By~\cite[Theorem~6.1]{KM16}, \((\hW^\eps, \hbW^\eps)\) converges weakly to
\((W, \bW)\) in \(C^0([0,\infty),\bR^{Nd} \times \bR^{Nd \times Nd})\),
where \(W\) is the \(N\)-dimensional Brownian motion with covariance 
matrix \(\Sigma = \tSigma\),
and
\(
  \bW (t) = \int_{0}^{t} W \otimes dW + E t
\).
Converting the It\^o integral to Stratonovich, we obtain
\[
  \bW (t) = \int_{0}^{t} W \otimes \circ dW - \frac{1}{2} \Sigma t + E t
  = \int_{0}^{t} W \otimes \circ dW
  .
\]

We modelled the flow \(F^t\) as a suspension 
over \(F\) with the roof function identically equal to \(1\).
Both \(h\) and \(V\) are bounded observables,
and we have the \(L^\infty\) martingale-coboundary decomposition
\(V = m + \chi \circ F - \chi\). Therefore the results 
of~\cite[Subsection~7.2]{KM16} apply. By~\cite[Corollary~7.6]{KM16},
for every \(q > 3\) there exists \(C > 0\) such that
for all \(0 \leq s \leq t\),
\begin{equation}
  \label{eq:mmll}
  \|\hW^\eps (s,t)\|_{2q} \leq C |t-s|^{1/2}
  \qquad \text{and} \qquad
  \|\hbW^\eps (s,t)\|_q \leq C |t-s|.
\end{equation}

\subsection{Completion of the proof of Lemma~\ref{lm:prer}}

We proved Lemma~\ref{lm:prer} for the processes \((\hW^\eps, \hbW^\eps)\)
in place of \((W^\eps, \bW^\eps)\). As in Remark~\ref{rmk:wic}, the two are
related by the time shift
\[
  \hW^\eps(t) = W^\eps (\eps^2 \tau, \eps^2 \tau + t)
  \qquad \text{and} \qquad
  \hbW^\eps(t) = \bW^\eps (\eps^2 \tau, \eps^2 \tau + t)
\]
for all \(t \geq 0\).

It remains to prove the moment bounds~\eqref{eq:anya}, based on~\eqref{eq:mmll},
and to show that \((\hW^\eps, \hbW^\eps)\) and \((W^\eps, \bW^\eps)\)
are close in \(C^0([0,\infty),\bR^{Nd} \times \bR^{Nd \times Nd})\).
This is done in the following two propositions.

\begin{prop}
  For every \(q > 3\) there exists \(C > 0\) such that
  for all \(0 \leq s \leq t\),
  \begin{align}
    \|W^\eps (s,t)\|_{2q} &\leq C |t-s|^{1/2} \label{eq:lmdpc}
    \\
    \|\bW^\eps (s,t)\|_q &\leq C |t-s|  \label{eq:lmdpd}
  \end{align}
\end{prop}

\begin{proof}
  As in Remark~\ref{rmk:wic}, we assume that \(\tau\), \(W^\eps\) and 
  \(\hW^\eps\) are defined on the same probability space such that
  \(
    W^\eps(\eps^2 \tau + t) - W^\eps(\eps^2 \tau)
    = \hW^\eps (t)
  \)
  for all \(t \geq 0\),
  and \(\hW^\eps\) is independent from \(\tau\).
  
  Note that \(\|\tau\|_q\) is finite for every \(q \geq 1\). 
  
  Assume that \(0 \leq s \leq t\).
  First we show~\eqref{eq:lmdpc}.
  We consider three cases:
  \begin{itemize}
    \item[(a)] If \(s \leq t \leq \eps^2 \tau\), then
      \(|W^\eps(s,t) |\ll \eps^{-1} |t-s| \ll \tau^{1/2} |t-s|^{1/2}\),
      thus
      \[
        \| W^\eps(s,t) 1_{s \leq t \leq \eps^2 \tau}\|_{2q} \ll |t-s|^{1/2}
        .
      \]
    \item[(b)] If \(\eps^2 \tau \leq s \leq t\), then
      \(W^\eps(s,t) = \hW^\eps(s-\eps^2 \tau, t-\eps^2 \tau)\),
      and by~\eqref{eq:mmll},
      \[
        \|W^\eps(s,t) 1_{\eps^2 \tau \leq s \leq t}\|_{2q} \ll |t-s|^{1/2}
        .
      \]
    \item[(c)] If \(s \leq \eps^2 \tau \leq t\), then
      \(W^\eps(s,t) = W^\eps(s,\eps^2 \tau) + W^\eps(\eps^2 \tau, t)\) and
      by (a) and (b),
      \begin{align*}
        \|W^\eps(s,t) 1_{s \leq \eps^2 \tau \leq t} \|_{2q}
        & \leq
          \|W^\eps(s,\eps^2 \tau ) 1_{s \leq \eps^2 \tau \leq t} \|_{2q}
        + \|W^\eps(\eps^2 \tau, t) 1_{s \leq \eps^2 \tau \leq t} \|_{2q}
        \ll |t-s|^{1/2}
        .
      \end{align*}
  \end{itemize}
  The bound~\eqref{eq:lmdpc} follows from the above. 
    
  The bound~\eqref{eq:lmdpd} is proved similarly, we consider the same three
  cases:
  \begin{itemize}
    \item[(a)] Suppose that \(s \leq t \leq \eps^2 \tau\).
      Observe that the variation of \(W^\eps\) on the interval \((t,s)\) is
      \(O(\eps^{-1} |t-s|)\), and \(|W^\eps(s,r)| \ll \eps^{-1} |t-s|\) for \(s \leq r \leq t\).
      Then
      \[
        | \bW^\eps(s,t) |
        = \Bigl| \int_s^t W^\eps(s,r) \otimes dW^\eps(r) \Bigr|
        \ll \eps^{-2} |t-s|^2
        \leq \tau |t-s|
        .
      \]
      Thus
      \(
        \|\bW^\eps(s,t) 1_{s \leq t \leq \eps^2 \tau} \|_q
        \ll \|\tau\|_q |t-s| \ll |t-s|
        .
      \)
    \item[(b)] If \(\eps^2 \tau \leq s \leq t\), then
      \(\bW^\eps(s,t) = \hbW^\eps(s-\eps^2 \tau, t-\eps^2 \tau)\), so
      \(
        \|\bW^\eps(s,t) 1_{\eps^2 \tau \leq s \leq t}\|_q
        \ll |t-s|
      \)
      by~\eqref{eq:mmll}.
    \item[(c)] Suppose that \(s \leq \eps^2 \tau \leq t\). It is convenient to use Chen's relation
      \[
        \bW^\eps(s,t) - \bW^\eps(s,\eps^2 \tau) - \bW^\eps(\eps^2 \tau,t) = W^\eps(s,\eps^2 \tau) \otimes W^\eps(\eps^2 \tau,t)
        .
      \]
      Estimate
      \begin{align*}
        \|\bW^\eps(s,\eps^2 \tau) 1_{s \leq \eps^2 \tau \leq t} \|_q & \ll |t-s|
        && \text{by (a),} \\
        \|\bW^\eps(\eps^2 \tau,t) 1_{s \leq \eps^2 \tau \leq t} \|_q & \ll |t-s|
        && \text{by (b),} \\
        \| W^\eps(s,\eps^2 \tau) \otimes W^\eps(\eps^2 \tau,t) 1_{s \leq \eps^2 \tau \leq t} \|_q & \ll |t-s|
        && \text{by H\"older's inequality and \eqref{eq:lmdpc}.}
      \end{align*}
      The bound \(\|\bW^\eps(s,t) 1_{s \leq \eps^2 \tau \leq t} \|_q \ll |t-s|\) follows.
  \end{itemize}
  The relation~\eqref{eq:lmdpd} follows. The proof is complete.
\end{proof}

\begin{prop}
  There exists \(C > 0\) such that for all $t \geq 0$,
  \begin{itemize}
    \item[(a)] \(\bigl| W^\eps(t) - \hW^\eps(t) \bigr| \leq C \eps \tau \),
    \item[(b)] \(\bigl| \bW^\eps(t) - \hbW^\eps(t) \bigr| \leq C \eps^2 \tau^2\).
  \end{itemize}
  As a consequence, the processes \((W^\eps, \bW^\eps)\) and \((\hW^\eps, \hbW^\eps)\)
  converge to the same limit in $C^0$.
\end{prop}

\begin{proof}
  By construction of \(W^\eps(s,t)\) and \(\bW^\eps(s,t)\), for all $0 \leq s \leq t$,
  \[
    |W^\eps(s,t)| \ll \eps^{-1} |t-s|
    \qquad \text{and} \qquad
    |\bW^\eps(s,t)| \ll \eps^{-2} |t-s|^2
    .
  \]
  Similar bounds hold for \(\hW^\eps\) and \(\hbW^\eps\).
  The result for \(t \leq \eps^2 \tau\) follows directly from the above.
  Suppose that \(t \geq \eps^2 \tau\). Then
  \[
    \bigl| W^\eps(t + \eps^2 \tau) - \hW^\eps(t) \bigr|
    = \bigl| W^\eps(\eps^2 \tau) \bigr|
    \ll \eps \tau
    .
  \]
  By Chen's relation,
  \[
    \bigl| \bW^\eps(t + \eps^2 \tau) - \hbW^\eps(t) \bigr|
    = \bigl| \bW^\eps(\eps^2 \tau) 
      + W^\eps(\eps^2 \tau) \otimes W^\eps(\eps^2 \tau, t + \eps^2 \tau)
    \bigr|
    \ll \eps^2 \tau^2
    .
  \]
  Similarly one shows that
  \[
    \bigl|W^\eps(t) - W^\eps(t + \eps^2 \tau)\bigr|
    \ll \eps \tau
    \qquad \text{and} \qquad
    \bigl|\bW^\eps(t) - \bW^\eps(t + \eps^2 \tau)\bigr|
    \ll \eps^2 \tau^2
    .
  \]
  The result follows.
\end{proof}

\section{A Heuristic Analysis}
\label{sec:RPV}

In this section we give a heuristic derivation of our result based on 
studying directly the distribution function. This derivation works also for 
more general collision models, see~\cite{BCKLs}. To extend the present result 
to those models requires the analysis with noise that is not independent from 
the slow variables. This is still a vastly unexplored area.

As before, we consider only the distribution of velocities, which is 
independent of the positions, and the electric field is $E = \eps \hn$.
Let \(p = (p_1, \ldots, p_N)\) denote velocities of the $N$ particles and
let \(F_t^\eps(p)\) be the density of the velocity distribution 
at time \(t\). From~\eqref{eq:gt0} we get
\begin{align*}
  \partial_t F_t^\eps(p)
  & + \sum_i \nabla_{p_i} \Bigl(
     \Bigl[E-\frac{\sum_k E \cdot p_k}{U} p_i\Bigr]F_t^\eps(p)
  \Bigr)
  \\ & = \lambda \sum_i\int_{S^{d-1}}
    \bigl(F_t^{\eps}(p_1, \ldots, |p_i|\omega,\ldots,p_N)-F_t^{\eps}(p)\bigr)
    \, d\sigma(\omega)
    ,
\end{align*}
where $\sigma(\omega)$ is the normalized volume measure on $S^{d-1}$. We 
write the above equation as
\[
 \partial_t F_t^\eps+\eps{\cB} F_t^\eps=\cA F_t^\eps.
\]
Rescaling time as $\tF_t^\eps(v)=F_{t/\eps^2}^\eps(v)$ and assuming that 
\[
  \tF_t^\eps(p)=\tF_t^0(p)+\eps \widetilde 
  F_t^{(1)}(p)+\eps^2 \tF_t^{(2)}(p) + o(\eps^2)
  ,
\]
we get, collecting the coefficients of powers of \(\eps\),
\begin{align}
  0 & = \cA \tF_t^0  \label{one}
  \\
  \cB \tF_t^0 
  & = \cA \tF_t^{(1)}\label{two}
  \\
  \dot {\tF_t^0}- \cB \tF_t^{(1)} 
  & = \cA \tF_t^{(2)} \label{three}
\end{align}
where the dot indicates differentiation with respect to $t$.
From \eqref{one} it follows that $\tF_t^0$ depends only on $v_k$ 
while substituting \eqref{two} into \eqref{three} gives 
\begin{equation}
  \label{formal}
  \dot {\tF_t^0}=P^\perp \cB \cA^{-1} \cB \tF_t^0
  ,
\end{equation}
where $P^\perp$ is the orthogonal projection from $L^2(S^{Nd-1})$ to the kernel 
$H_0$ of $\cA$. Observe that $\cA^{-1} \cB$ is well defined: the image 
of $\cB$ is contained in $H_0^\perp$. Writing \eqref{formal} explicitly 
gives

\begin{equation}\label{mas}
  \dot{\tF_t^0}(v)
  = - \delta \sum_i 
    \frac{\partial}{\partial v_i}
    \Bigl(
      \Bigl[ \frac{d-1}{2v_i}-\frac{(Nd-1)v_i}{U} \Bigr]
      \tF_t^{0}(v)
    \Bigr)
    + \frac{\delta}{2} \sum_{i,j} 
      \frac{\partial^2}{\partial v_i \partial v_j}
      \Bigl(
        \Bigl[\delta_{i,j}-\frac{v_i v_j}{U}\Bigr]
        \tF^0_t(v)
      \Bigr)
  ,
\end{equation}
where $\delta=2\lambda^{-1} d^{-1}$. Equation~\eqref{mas} is the Master Equation for 
the SDE \eqref{SDE}. 

Let $f^{\eps,N}_t(p_1)$ be the one particle marginal of $F_t^\eps(p)$,
\[
 f^{\eps,N}_t(p_1)=\int F_t^\eps(p_1,p_2,\ldots,p_N) \, dp_2 \cdots dp_N
 .
\]
In \cite{BCELM} it was shown that the limit
\[
 f^{\eps}_t(p)=\lim_{N\to\infty} f^{\eps,N}_t(p_1)
\]
satisfies the Boltzmann-Vlasov equation
\[
  \dot{f}^{\eps}_t(p)
  + \nabla_{p} \Bigl(
    \Bigl[E-\frac{E j(t)}{U}p\Bigr]
    f_t^\eps(p)
  \Bigr)
  = \lambda^{-1}
  \int_{S^{d-1}} \bigl( f_t^{\eps}(|p|\omega) - f_t^{\eps}(p) \bigr)
  \, d\sigma(\omega)
  ,
\]
where $j(t)$ is fixed by the \emph{self-consistent} condition
\[
 j(t)=\int_{\bR^d} p f^\eps_t(p)\, dp
 .
\]
We can repeat the above scaling analysis by setting
$\tf^{\eps}_t(p)=f^{\eps}_{t/\eps^2}(p)$ and taking the van Hove limit
\[
 \tf^0_t(p)=\lim_{\eps\to 0} f^\eps_t(p)
 .
\]
A formal perturbative argument very similar to the one used for $F_t^\eps(p)$ 
gives that $\tf^0(p)$ depends only on $v=|p|$ and that it satisfies
\begin{equation}\label{OU}
  \dot{\tf}^0_t(v)
  =\delta \frac{d}{dv} (v\tf^0_t(v))
  + \frac{\delta}{2} \frac{d^2}{dv^2} \tf^0_t(v)
  \, .
\end{equation}
In Appendix~\ref{sec:PSD} we show that for large $N$,
the speed of an individual particle in~\eqref{SDE} is close to 
an Ornstein-Uhlenbeck process whose Fokker-Planck equation is~\eqref{OU}. Thus 
taking the van Hove scaling $E\to 0$ and then the large system limit 
$N\to\infty$ is, at least formally, equivalent to taking the large system limit 
before the van Hove scaling. In this sense the van Hove scaling studied in this 
paper is consistent, at least at a formal level, with the large $N$ limit 
studied in \cite{BCELM}.

\appendix

\section{Continuity of solution map of differential equations}
  \label{sec:CDESM}

  Suppose that $x(t)$, $0 \leq t \leq 1$, is a continuously differentiable path in $\bR^2$,
  and that $A \colon \bR^2 \to \bR^{2 \times 2}$ is a smooth matrix-valued function.
  Let $y$ be a solution of an integral equation
  \begin{equation}
    \label{eq:mmaa}
    y (t) = \int_0^t A(y(s)) \, dx(s)
    .
  \end{equation}
  The integral above is understood in the Riemann-Stieltjes sense,
  and $y$ is uniquely defined.

  Let \(\Gamma \colon C^1([0,1], \bR^2]) \to C^0([0,1], \bR^2)\) be the 
  \emph{solution map} for~\eqref{eq:mmaa}. That is,
  \(\Gamma(x) = y\).
  
  It follows from Gr\"onwall's inequality that \(\Gamma\) is continuous.
  So, if a sequence \(x^\eps\) converges to \(x^0\) in \(C^1\) topology as \(\eps \to 0\),
  then the corresponding sequence \(y^\eps = \Gamma (x^\eps)\) converges
  to \(y^0 = \Gamma(x^0)\) in the \(C^0\) topology.
  
  The domain of \(\Gamma\) can be extended to the space $C^\alpha$ of
  \(\alpha\)-H\"older paths when \(\alpha > 1/2\).
  (Or alternatively to the space of paths of bounded $p$-variation with $p<2$.)
  In this case, the integral in~\eqref{eq:mmaa} is a Young integral~\cite{Young36}.
  The map $\Gamma$ is still continuous on $C^\alpha$, see 
  \cite[Theorem~1.28]{Lyons07}, \cite[Section 8.6]{FrizHairer14}.
  
  But sample paths of Brownian motions are \(\alpha\)-H\"older continuous only when
  \(\alpha < 1/2\), where it is problematic to extend $\Gamma$ in a meaningful way.
  We illustrate a problem with continuity of possible extensions of $\Gamma$
  by the following standard example.
  
  Let
  \[
    A \colon \begin{pmatrix} a \\ b \end{pmatrix}
    \mapsto \begin{pmatrix} 1 & 0 \\ 0 & a \end{pmatrix}
    .
  \]
  Then $y_1 = x_1$ and $y_2(t) = \int_0^t x_1(s) \dot{x}_2(s) \, ds$.
  Let $x^\eps$, $\eps > 0$, be sequence of smooth paths
  \[
    x^\eps (t) = \eps 
    \begin{pmatrix} 
      \cos (\eps^{-2} t) \\ 
      \sin (\eps^{-2} t)
    \end{pmatrix}
    .
  \] 
  It is easy to see that  \(x^\eps\) converges to \(x^0 \equiv 0\) in 
  \(\alpha\)-H\"older topology for each \(\alpha < 1/2\) (but \textbf{not} for \(\alpha \geq 1/2\)). 
  For small \(\eps\),
  \[
    y^\eps(t) = \begin{pmatrix}
      \eps \cos(\eps^{-2} t) \\
      \int_0^t \cos^2 (\eps^{-2} s) \, ds
    \end{pmatrix}
    \approx
    \begin{pmatrix}
      0 \\
      t/2
    \end{pmatrix}
    .
  \]
  Hence \(y^\eps\) does not converge to \(y^0 = \Gamma(x^0) \equiv 0\).
  Thus \(\Gamma\) cannot be extended to a continuous map on
  the space of \(\alpha\)-H\"older paths, \(\alpha < 1/2\).

\begin{rmk}
  In fact, there is no separable Banach space \(\cB \subset C^0([0,\infty), \bR^2)\)
  such that:
  \begin{itemize}
    \item sample paths of Brownian motions lie in \(\cB\) almost surely,
    \item the map \(\Gamma\), defined on smooth paths,
      extends to a continuous map \(\Gamma \colon \cB \to C^0([0,\infty), \bR^2)\).
  \end{itemize}
  See~\cite{Lyons91,Lyons07} or~\cite[Proposition~1.1]{FrizHairer14} for details.
\end{rmk}

\section{Projections of spherical diffusion}
\label{sec:PSD}

For each $n \geq 1$, suppose that $W$ is a Brownian motion in $\bR^n$
with identity covariance matrix.
Define a stochastic process $u$ in $\bR^n$ as a solution
of the Stratonovich differential equation
\[
  du = dW - u \frac{u^* \circ dW}{n}
  , \qquad u(0) = \xi
  .
\]
We require that $\xi$ belongs to the sphere
$S = \{x \in \bR^n : |x| = n\}$.
Then $u$ is a diffusion on $S$.

We are interested in statistical behavior of the one dimensional projections of $u$,
say the first coordinate $u_1$, with large $n$.
For the initial condition $\xi$, we fix $\xi_1$ independent of $n$
and choose $\xi_j$, $j \geq 2$, arbitrarily, deterministic or random
independent of $W$.

\begin{thm}
  As $n \to \infty$, the process $u_1(t)$ converges weakly
  (in the uniform topology)
  to an Ornstein-Uhnelbeck process
  \[
    dX = dB - \frac{1}{2} X \, dt
    , \qquad X(0) = \xi_1
    ,
  \]
  where $B$ is a standard one-dimensional Brownian motion.
\end{thm}

\begin{proof}
  We write the stochastic differential equation for $u$ in the It\^o
  form~\cite[pages 137--138]{Blowey01}:
  \[
    du = dW - u \frac{u^* \, dW}{n} - \frac{n-1}{2n} u \, dt
    .
  \]
  Denote $p = u_1$. Then
  \[
    dp = \Bigl(1-\frac{p^2}{n}\Bigr) dW_1
    - \frac{\sqrt{n-p^2}}{n} dW' - \frac{n-1}{2n} p \, dt
    , 
  \]
  where $W'$ is a standard one-dimensional Brownian motion,
  independent from $W_1$, which appears as
  \[
    dW' 
    = \frac{\sum_{j \geq 2} u_j \, dW_j}{\sqrt{\sum_{j \geq 2} u_j^2}}
    = \frac{\sum_{j \geq 2} u_j \, dW_j}{\sqrt{n-p^2}}
    .
  \]
  Let $X$ be the Ornstein-Uhlenbeck process,
  $dX = dW_1 - \frac{1}{2} X \, dt$, $X(0) = \xi_1$.
  Let $\delta = p - X$. Note that $\delta(0) = 0$.
  We will show that $\delta(t)$ remains \emph{small}
  for $t \geq 0$.
  Write
  \begin{align*}
    d\delta
    & = \frac{-p^2}{n}\, dW_1 - \frac{\sqrt{n-p^2}}{n} \, dW'
      - \frac{\delta}{2} \, dt + \frac{p}{2n} \, dt
    \\ & = \frac{1}{\sqrt{n}} \, dW'' - \frac{\delta}{2} \, dt
    + \frac{p}{2n} \, dt
    ,
  \end{align*}
  where $W''$ is a standard Brownian motion.
  Another way of writing the above is
  \[
    \delta(t) 
    = \frac{1}{\sqrt{n}} W''(t)
    + \int_0^t \frac{p(s)}{2n} \, ds
    - \frac{1}{2} \int_0^t \delta(s) \, ds
    .
  \]
  By construction, $|p| \leq \sqrt{n}$ at all times. Let 
  \[ 
    \alpha(t) = |\delta(t)|
    \qquad \text{and} \qquad
    \beta(t) = \frac{|W''(t)|}{\sqrt{n}} + \frac{t}{2\sqrt{n}}
    .
  \]
  Then
  \[
    0 \leq \alpha(t) 
    \leq \beta(t) + \frac{1}{2} \int_0^t \alpha(s) \, ds
    .
  \]
  By the Gronwall inequality \cite[Lemma 4.5.1]{KloedenPlaten92},
  \[
    \alpha(t)
    \leq \beta(t) + \frac{1}{2} \int_0^t e^{(t-s)/2} \beta(s) \, ds
    \leq \Bigl(1 + \frac{t e^t}{2}\Bigr) \hat{\beta}(s)
    ,
  \]
  where $\hat{\beta}(t) = \max_{s \in [0,t]} \beta(s)$.
  Let also $\hat{W}''(t) = \max_{s \in [0,t]} |W''(s)|$.
  By Burkholder's inequality,
  $\bE \hat{W}''(t) \leq C \sqrt{t}$, where $C$ is an absolute
  constant. Thus
  $\bE \hat{\beta}(t) \leq C (\sqrt{t} + t) / \sqrt{n}$.
  
  We have shown that $p = X + \delta$, where $X$ is the required
  Ornstein-Uhlenbeck process, and
  $
    \bE (\sup_{s \leq t} |\delta(s)|)
    \leq C_t / \sqrt{n}
    ,
  $
  where $C_t > 0$ only depends on $t$.
  This implies our result. 
\end{proof}

\subsection*{Acknowledgements}

This material is based upon work supported by AFOSR under the award number
FA9500-16-1-0037.

A.K.\ was funded by a European Advanced Grant {\em StochExtHomog} (ERC AdG 320977) at the University of Warwick
and an Engineering and Physical Sciences Research Council grant EP/P034489/1 at the University of Exeter. 

F.B. gratefully acknowledges financial support from the Simons Foundation award 
number 359963.

\end{document}